\documentclass[letterpaper, 10 pt, conference]{ieeeconf}  % Comment this line out if you need a4paper

\IEEEoverridecommandlockouts                              % This command is only needed if 
\usepackage{graphicx}
\usepackage{subcaption}

\let\labelindent\relax
\usepackage{sty/nlab-settings}
\usepackage{sty/nlab-commands}

\usepackage{sty/nlab-theorems}

\usepackage{mathtools}
\mathtoolsset{showonlyrefs}
\title{\LARGE \bf
Optimization-Free Constrained Control with Guaranteed Recursive Feasibility: A CBF-Based Reference Governor Approach
}

\author{Satoshi Nakano$^{1}$, Emanuele Garone$^{2}$, and Gennaro Notomista$^{3}$% <-this % stops a space
\thanks{*This work was supported by JSPS KAKENHI Grant Numbers JP23K13352.}% <-this % stops a space
\thanks{$^{1}$Satoshi Nakano is with the Department of Engineering,
        Nagoya Institute of Technology, 466-8555 Aichi, Japan
        {\tt\small nakano@nitech.ac.jp}}%
\thanks{$^{2}$Emanuele Garone is with the Department of Control Engineering and System Analysis,
        Universite Libre de Bruxelles, 1050 Brussels, Belgium
        {\tt\small egarone@ulb.ac.be}}%
\thanks{$^{3}$Gennaro Notomista is with the Department of Electrical and Computer Engineering, University of Waterloo,
        ON N2L 3G1, Canada
        {\tt\small gennaro.notomista@uwaterloo.ca}}%
}

\begin{document}

\maketitle
\thispagestyle{empty}
\pagestyle{empty}

%%%%%%%%%%%%%%%%%%%%%%%%%%%%%%%%%%%%%%%%%%%%%%%%%%%%%%%%%%%%%%%%%%%%%%%%%%%%%%%%
\begin{abstract}
    This letter presents a constrained control framework that integrates Explicit Reference Governors (ERG) with Control Barrier Functions (CBF) to ensure recursive feasibility without online optimization.
    We formulate the reference update as a virtual control input for an augmented system, governed by a smooth barrier function constructed from the softmin aggregation of Dynamic Safety Margins (DSMs).
    Unlike standard CBF formulations, the proposed method guarantees the feasibility of safety constraints by design, exploiting the forward invariance properties of the underlying Lyapunov level sets.
    This allows for the derivation of an explicit, closed-form reference update law that strictly enforces safety while minimizing deviation from a nominal reference trajectory.
    Theoretical results confirm asymptotic convergence, and numerical simulations demonstrate that the proposed method achieves performance comparable to traditional ERG frameworks.
\end{abstract}

\section{INTRODUCTION}
Safety-critical control has emerged as a fundamental requirement in autonomous systems, where state and input constraints must be strictly satisfied during operation \cite{amesControlBarrierFunctions2019}.
Among various approaches, Control Barrier Functions (CBFs) have gained significant popularity due to their ability to enforce forward invariance of safe sets via low-complexity constraints on the control input \cite{amesControlBarrierFunctions2019,9131819}.
Typically, CBFs are incorporated into a Quadratic Program (QP) that acts as a safety filter for a nominal controller.
However, a major theoretical and practical limitation of standard CBF-QP-based methods is the issue of feasibility.
Since the control input is computed myopically, there is no inherent guarantee that the QP remains feasible for all time, especially when the system is subject to tight input bounds or when multiple barrier constraints conflict.
Although various techniques have been proposed to enhance feasibility, such as backup sets or adaptive coefficients, they often introduce additional computational complexity or require heuristic tuning \cite{chen2021,xiao2022}.

An alternative strategy for constrained control is the Reference Governor (RG) framework, which manages constraints by modifying the reference signal supplied to a pre-stabilized closed-loop system \cite{garone2017}.
The Explicit Reference Governor (ERG) \cite{7244340,8412335} is a computationally efficient formulation of RGs that avoids online optimization by utilizing the concept of Dynamic Safety Margin (DSM).
ERG has been applied to a variety of applications \cite{nakano2023,tartaglione2024,goldar2021}.
The DSM is a scalar value derived from a Lyapunov function that quantifies the available margin before constraint violation occurs.
By modulating the speed of the reference dynamics based on the DSM, the ERG ensures that the system state never leaves a pre-calculated invariant set, thereby guaranteeing recursive feasibility by design.
However, traditional ERG schemes typically employ a static field to update the reference.
Designing this field to prevent stagnation near constraints in complex environments is non-trivial and often requires heuristic design of repulsion terms.

In this letter, we propose a novel ERG formulation that integrates the recursive feasibility guarantees of the ERG with the CBF-guided reference update.
Specifically, we explicitly formulate the time derivative of the auxiliary reference as a virtual control input, which enables the direct application of CBF-based design to the reference dynamics.
We construct a unified barrier function by aggregating the DSMs of all constraints and the steady-state admissibility conditions using a softmin function \cite{rabiee2024a}.
This formulation allows us to derive a reference update law based on the CBF condition for the augmented system.
Crucially, because the underlying safety quantification relies on the DSM level sets, which are forward invariant for the nominal controller, the resulting constraints on the reference update are always feasible (e.g., by maintaining a constant reference).
We provide a closed-form solution for the reference update that minimizes the deviation from a nominal gradient flow while strictly enforcing safety.

The remainder of this letter is organized as follows.
Sections \ref{sec:problem} and \ref{sec:preliminaries} formulate the control problem and introduce the necessary preliminaries, respectively.
Section \ref{sec:reference_governor_design} details the proposed ERG-CBF framework, including the derivation of the closed-form reference update law.
Section \ref{sec:sim} presents numerical simulations, and Section \ref{sec:conclusion} concludes the letter.

\section{PROBLEM FORMULATION}
\label{sec:problem}
Consider the nonlinear control system given by
\begin{align}
    \dot{x}(t) = f(x(t),u(t)), \label{eq:system}
\end{align}
where $x(t) \in \nlR^n$ is the state and $u(t) \in \nlR^m$ is the control input.
Let $r \in \nlR^\ell$ denote the constant desired reference.
We assume that a nominal controller
\begin{align}
    u(t) = \kappa(x(t), g(t))
    \label{eq:controller}
\end{align}
has been predesigned, where $g(t) \in \nlR^\ell$ is an auxiliary reference.
We assume that $f$ and the prestabilizing controller $\kappa$ are $C^1$.
For any constant reference $g(t)\equiv\bar g$, the closed-loop system admits a unique asymptotically stable equilibrium $x_g \in \nlR^n$ satisfying $f(x_g,\kappa(x_g,\bar g))=0$, and the mapping $g \mapsto x_g$ is continuously differentiable.
That is, for any initial condition in the region of attraction,
\begin{align}
    \lim_{t \to \infty} \|x(t) - x_g\| = 0.
\end{align}
The system is subject to $n_c$ constraints defined by
\begin{align}
    h_i(x(t), g(t)) \ge 0, \quad i = 1, \dots, n_c,
\end{align}
where $h_i : \nlR^n \times \nlR^\ell \to \nlR$ are continuously differentiable functions.
These constraints may represent physical limitations on the state, actuator saturation (via the dependence on $g$ in the control law), or safety requirements.

The control objective is to design a reference governor that manipulates the auxiliary reference $g(t)$ to steer the state $x(t)$ to the equilibrium $x_r$ associated with $r$, while ensuring that the constraints $h_i(x(t), g(t)) \ge 0$ are satisfied for all $t \ge 0$.

\section{PRELIMINARIES}
\label{sec:preliminaries}

\subsection{Control Barrier Functions and Softmin Aggregation}

Consider a continuously differentiable function $h : \nlR^n \to \nlR$ and the associated safe set $\mathcal{S} \coloneqq \{ x \in \nlR^n \mid h(x) \ge 0 \}$.
The set $\mathcal{S}$ is said to be \emph{forward invariant} for the closed-loop system if $x(0)\in\mathcal{S}$ implies $x(t)\in\mathcal{S}$ for all $t\ge0$.
The function $h$ is a \emph{Control Barrier Function} (CBF) if there exists an extended class-$\mathcal{K}$ function $\alpha$ such that
\begin{align}
    \dot h(x(t)) + \alpha(h(x(t))) \ge 0
    \label{eq:CBF_continuous}
\end{align}
holds along the trajectories of the closed-loop system~\cite{amesControlBarrierFunctions2019,9131819}.
For the closed-loop system $\dot x = f(x,\kappa(x,g))$, the CBF condition reads
\begin{align}
    \nabla_x h(x)^\nlT f(x,\kappa(x,g)) + \alpha(h(x)) \ge 0.
    \label{eq:CBF_condition}
\end{align}
If \eqref{eq:CBF_condition} holds for all $t\ge0$, then $\mathcal{S}$ is forward invariant.

When multiple constraints $h_i(x)\ge0$, $i\in\mathcal{I}$, are present, the safe set is the intersection $\mathcal{S}=\bigcap_{i\in\mathcal I}\{x\mid h_i(x)\ge0\}$.
To obtain a smooth inner approximation of this intersection, we aggregate the individual barrier candidates via the \emph{soft minimum} (log-sum-exp) \cite{rabiee2024a,rabiee2025}.
Define the softmin operator for a collection $\{s_i\}$ and a parameter $\beta>0$ by
\begin{align}
    \operatorname{softmin}_\beta\{s_i\}
    \coloneqq
    -\frac{1}{\beta}\log\!\Big(\sum_i e^{-\beta s_i}\Big).
\end{align}
We then set, with the aggregation parameter $\beta_H>0$,
\begin{align}
    H(x) \coloneqq \operatorname{softmin}_{\beta_H}\{h_i(x)\}_{i\in\mathcal I}.
\end{align}
The function $H$ is $C^\infty$ whenever the $h_i$ are present and satisfies
$\min_i h_i(x) - \frac{1}{\beta_H}\log n_c \le H(x) \le \min_i h_i(x)$,
with $\lim_{\beta_H\to\infty} H(x)=\min_i h_i(x)$.
Since $H$ is continuously differentiable whenever the $h_i$ are, it can be treated as a candidate CBF.
Imposing the single CBF condition $\dot H(x)+\alpha_H(H(x))\ge0$ guarantees forward invariance of the inner approximation $\{x\mid H(x)\ge0\}$ and hence enforces all $h_i(x)\ge0$.
In the sequel, this softmin construction will be naturally extended to functions depending on both the state $x$ and the auxiliary reference $g$.

\subsection{Dynamic Safety Margins}
Dynamic Safety Margins (DSMs) provide a Lyapunov-based measure of the distance to constraint violation for prestabilized systems and play a central role in the Explicit Reference Governor (ERG)~\cite{8412335,11006078,freire2025a}.
In this work, we employ the DSM to ensure the transient safety of the closed-loop trajectories, independent of the steady-state admissibility conditions that will be introduced later.

\begin{defn}[Reference-dependent Lyapunov function~{\cite{11006078,freire2025a}}]\label{def:ref_dep_V}
    For a constant reference $g \in \nlR^\ell$, let $x_g$ denote the corresponding equilibrium of the closed-loop system under the controller $\kappa(x,g)$.
    A continuously differentiable function $V : \nlR^n \times \nlR^\ell \to \nlR$ is called a \emph{reference-dependent Lyapunov function} if, for each $g$, there exists a neighborhood $D_g \subset \nlR^n$ of $x_g$ such that
    \begin{subequations}
        \begin{align}
            V(x_g, g)                                         & = 0,                                              \\
            V(x, g)                                           & > 0, \quad \forall x \in D_g \setminus \{x_g\}, \\
            \frac{\partial V}{\partial x}(x,g)   f(x,\kappa(x,g)) & \le 0, \quad \forall x \in D_g.
        \end{align}
    \end{subequations}
\end{defn}

We assume that for any fixed $g$, $V(x,g)$ is radially unbounded with respect to $x$, ensuring that the sublevel sets $\{x \mid V(x,g) \le c\}$ are compact.
For each constraint $h_i$, we define the constraint-wise admissible region $\mathcal{S}_{i,g}$ and its complement $\mathcal{S}_{i,g}^c$ within the domain $D_g$, and the associated thresholds
\begin{align}
    \Gamma_i^\ast(g) &\coloneqq \inf_{x \in \mathcal{S}_{i,g}^c} V(x,g) - \inf_{x \in \mathcal{S}_{i,g}} V(x,g), \\
    \bar{\Gamma}(g)  &\coloneqq \inf_{x \in \partial D_g} V(x,g).
\end{align}

Following the standard DSM construction~\cite{11006078,freire2025a}, define the candidate margins
\begin{align}
    m_1(x,g) &\coloneqq \Gamma_i^\ast(g) - V(x,g), \\
    m_2(x,g) &\coloneqq (1-\epsilon)\bar{\Gamma}(g) - V(x,g),
\end{align}
where $\epsilon \in (0, 1)$ is a small margin parameter and $m_2$ is relevant only when $V$ is valid in a restricted region of attraction.\footnote{If $V$ is globally valid, the stability margin is unnecessary and one may simply set $\Delta_i(x,g)=m_1(x,g)$.}

\begin{defn}[Dynamic Safety Margin for constraint $i$]\label{def:dsm}
For $\beta_\Delta>0$, the smooth DSM is defined by
\begin{align}
    \Delta_i(x,g)
    \coloneqq
    \operatorname{softmin}_{\beta_\Delta}\{m_1(x,g),\,m_2(x,g)\}.
\end{align}
\end{defn}

For finite $\beta_\Delta$, $\Delta_i$ is $C^\infty$ whenever $V$ is $C^\infty$, and $\lim_{\beta_\Delta\to\infty}\Delta_i(x,g)=\min\{m_1,m_2\}$ pointwise; hence, the classical DSM is recovered in the limit.
Note that both $m_1$ and $m_2$ depend on $x$ only through $-V(x,g)$; hence, the $x$-gradient of the smooth DSM satisfies
\begin{align}
    \nabla_x \Delta_i(x,g) \;=\; -\nabla_x V(x,g),
\end{align}
so $\nabla_x \Delta_i(x,g)^\nlT f(x,\kappa(x,g)) = -\nabla_x V(x,g)^\nlT f(x,\kappa(x,g)) \ge 0$, which is utilized in subsequent proofs.

By definition, the condition $\Delta_i(x,g) \ge 0$ implies that $V(x,g) \le \Gamma_i^\ast(g)$ (and within the stability region), which guarantees that $x$ remains in the safe set $\mathcal{S}_{i,g}$.
Since $\dot{V}(x,g) \le 0$ for a constant $g$, each set $\mathcal{C}_i \coloneqq \{(x,g) \mid \Delta_i(x,g) \ge 0\}$ is forward invariant under the nominal controller $\kappa(x,g)$.

\section{REFERENCE GOVERNOR DESIGN}
\label{sec:reference_governor_design}
In this section, we propose a reference update scheme that steers the applied reference $g$ toward the target $r$ while guaranteeing constraint satisfaction.
Building upon the framework of~\cite{11006078}, we treat the time derivative of the reference as a virtual control input $\rho \in \nlR^\ell$, i.e., $\dot{g} = \rho$.
Consequently, the augmented closed-loop dynamics are given by
\begin{equation}
    \label{eq:augmented_system}
    \begin{cases}
        \dot{x} & = f(x, \kappa(x,g)), \\
        \dot{g} & = \rho.
    \end{cases}
\end{equation}
The objective is to design $\rho$ such that the constraints are satisfied and $g$ converges to $r$.

\subsection{Steady-State Admissibility}
In addition to transient safety, the reference governor must ensure that the applied reference $g$ remains \emph{steady-state admissible}, meaning that the corresponding equilibrium $x_g$ satisfies the constraints.
\begin{defn}[Steady-state admissible reference]
    The set of steady-state admissible references is defined as
    \begin{align}
        \mathcal{V} \coloneqq \{ g \in \nlR^\ell \mid h_i(x_g, g) \ge 0, \quad \forall i \in \mathcal{I} \}.
    \end{align}
\end{defn}
To design the reference update law, we distinguish between constraints that are automatically satisfied at steady state and those that impose restrictions on $g$.
Let $\mathcal{I}_{\mathrm{SA}} \subseteq \mathcal{I}$ be the set of indices such that for all $i \in \mathcal{I}_{\mathrm{SA}}$, the condition $h_i(x_g, g) \ge 0$ holds for all $g \in \nlR^\ell$.
The complementary set $\mathcal{I}_{\mathrm{SA}}^c \coloneqq \mathcal{I} \setminus \mathcal{I}_{\mathrm{SA}}$ contains the indices of constraints that explicitly restrict the admissible reference set $\mathcal{V}$.
Consequently, the proposed reference governor must update $g$ such that:
\begin{enumerate}
    \item The DSM condition $\Delta_i(x,g) \ge 0$ holds for all $i \in \mathcal{I}$ (ensuring transient safety).
    \item The steady-state condition $h_j(x_g, g) \ge 0$ holds for all $j \in \mathcal{I}_{\mathrm{SA}}^c$ (ensuring $g \in \mathcal{V}$).
\end{enumerate}
These requirements are summarized as:
\begin{subequations}
    \label{eq:constraints-erg}
    \begin{align}
        \Delta_i(x,g) & \ge 0, &  & \forall i \in \mathcal{I}, \label{eq:constraints-dsm}                \\
        h_j(x_g,g)    & \ge 0, &  & \forall j \in \mathcal{I}_{\mathrm{SA}}^c. \label{eq:constraints-sa}
    \end{align}
\end{subequations}

\subsection{Gradient-Based Reference Update}
To drive the reference $g$ toward the target $r$, we consider a nominal gradient flow
\begin{align}
    \dot{g} = - \nabla_g V_g(g,r),
\end{align}
where $V_g : \nlR^\ell \times \nlR^\ell \to \nlR_{\ge 0}$ is a potential function satisfying $V_g(g,r) > 0$ for $g \ne r$ and $V_g(r,r) = 0$.
However, this update must be modified to respect the safety constraints \eqref{eq:constraints-erg}.
To handle these constraints in a unified manner, we define the softmin-composed barrier function:
\begin{align}
    H(x,g)
     & \coloneqq \operatorname{softmin}_{\beta_H} \Big( \{ \Delta_i(x,g) \}_{i \in \mathcal{I}} \cup \{ h_j(x_g,g) \}_{j \in \mathcal{I}_{\mathrm{SA}}^c} \Big) \nonumber \\
     & \!=\!
    -\!\frac{1}{\beta_H}
    \log \Bigg(\!
    \sum_{i \in \mathcal{I}}
    e^{-\beta_H \Delta_i(x,g)}
    \!+\!
    \sum_{j \in \mathcal{I}_{\mathrm{SA}}^c}
    e^{-\beta_H h_j(x_g,g)}
    \!\Bigg).
    \label{eq:H_softmin_alt}
\end{align}
Since each aggregated term is lower bounded by $H(x,g)$, the condition $H(x,g)\ge0$ implies \eqref{eq:constraints-erg}.
For finite $\beta_H$, the converse implication may not hold, which yields a conservative inner approximation of the simultaneous constraints.
Therefore, $H(x,g)\ge0$ serves as a smooth sufficient condition for the simultaneous satisfaction of all safety requirements.
Whenever $H$ is differentiable, we can enforce the forward invariance of $\{(x,g)\mid H(x,g)\ge0\}$ using the CBF condition for the augmented system \eqref{eq:augmented_system}:
\begin{align}
    \dot{H}(x,g) + \alpha_H(H(x,g)) \ge 0,
    \label{eq:H_CBF_condition}
\end{align}
where $\alpha_H$ is a class-$\mathcal{K}$ function.
Expanding the time derivative $\dot{H} = \frac{\partial H}{\partial x} \dot{x} + \frac{\partial H}{\partial g} \dot{g}$ yields the affine constraint on the virtual input $\rho$:
\begin{align}
    a_H(x,g)^\nlT \rho \le b_H(x,g),
    \label{eq:H_linear_constraint}
\end{align}
where $a_H(x,g) \coloneqq -\nabla_g H(x,g)$ and $b_H(x,g) \coloneqq \nabla_x H(x,g)^\nlT f(x,\kappa(x,g)) + \alpha_H(H(x,g))$.
The proposed reference update law is obtained by solving the following Quadratic Program (QP), which minimizes the deviation from the nominal direction subject to the safety constraint:
\begin{align}
    \rho^\star(x,g) = \arg \min_{\rho \in \nlR^\ell}
    \|\rho + \nabla_g V_g(g,r)\|^2
    \quad
    \text{s.t.}\quad
    a_H^\nlT \rho \le b_H.
    \label{eq:QP_reference_update}
\end{align}
The explicit solution is given by
\begin{align}
    \dot g = \rho^\star \coloneqq
    \begin{cases}
        - \nabla_g V_g(g,r)
        -
        \frac{
            \max \big\{0,
            a_H^\nlT \big(-\nabla_g V_g(g,r)\big)
            - b_H
            \big\}
        }{
            \|a_H\|^2
        } a_H,
        \\[1mm]
        \hspace{8mm}\text{if } \|a_H\|^2 > 0, \\[2mm]
        - \nabla_g V_g(g,r),
        \qquad\text{if } \|a_H\|^2 = 0.
    \end{cases}
    \label{eq:QP_solution}
\end{align}
This reference update law ensures that the reference $g$ evolves as close as possible to the nominal gradient descent direction while strictly satisfying the safety condition $H(x,g) \ge 0$.

\subsection{Safety and Convergence Analysis}
We now establish the safety and convergence properties of the proposed closed-loop system.
First, we confirm that the reference update optimization is feasible for any state in the safe set.

\begin{prop}[Feasibility of the trivial update]
    \label{prop:feasibility}
    For any $(x,g)$ such that $H(x,g) \ge 0$, the choice $\rho = 0$ satisfies the CBF constraint \eqref{eq:H_linear_constraint}, i.e., $b_H(x,g) \ge 0$.
\end{prop}
\begin{proof}
    Let $S(x,g) \coloneqq \sum_{i \in \mathcal{I}} e^{-\beta_H \Delta_i(x,g)} + \sum_{j \in \mathcal{I}_{\mathrm{SA}}^c} e^{-\beta_H h_j(x_g,g)}$.
    Since each exponential term is strictly positive, $S(x,g)>0$ holds, hence the weights are well-defined.
    Set $w_{\Delta_i}(x,g) \coloneqq e^{-\beta_H \Delta_i(x,g)}/S(x,g)$ and $w_{h_j}(x,g) \coloneqq e^{-\beta_H h_j(x_g,g)}/S(x,g)$.
    Then $w_{\Delta_i}(x,g)>0$ and $w_{h_j}(x,g)>0$ hold, and $\sum_{i \in \mathcal{I}} w_{\Delta_i}(x,g) + \sum_{j \in \mathcal{I}_{\mathrm{SA}}^c} w_{h_j}(x,g) = 1$.
    Moreover, the chain rule gives the convex-combination form
    \begin{align}
        \nabla_x H(x,g)
         & =
        \sum_{i \in \mathcal{I}} w_{\Delta_i}(x,g)\nabla_x \Delta_i(x,g) \notag \\
         & +
        \sum_{j \in \mathcal{I}_{\mathrm{SA}}^c} w_{h_j}(x,g)\nabla_x h_j(x_g,g).
        \label{eq:Hx_convex_combination}
    \end{align}
    For each DSM term, we established that $\nabla_x \Delta_i(x,g) = -\nabla_x V(x,g)$.
    Therefore, along the closed-loop dynamics with $\rho=0$, we obtain $\nabla_x \Delta_i(x,g)^\nlT f(x,\kappa(x,g)) = -\nabla_x V(x,g)^\nlT f(x,\kappa(x,g)) \ge 0$,
    where the inequality follows from the reference-dependent Lyapunov property $\dot V(x,g)\le0$ for constant $g$.
    Note that $\nabla_x h_j(x_g,g)=0$ since the dependence on $x$ is through the equilibrium $x_g(g)$.
    Combining these facts with \eqref{eq:Hx_convex_combination} yields $\nabla_x H(x,g)^\nlT f(x,\kappa(x,g)) \ge 0$.
    Since $\alpha_H(H(x,g))\ge0$ holds whenever $H(x,g)\ge0$, we conclude that $b_H(x,g)\ge0$ holds, and thus $\rho=0$ is feasible.
\end{proof}

\begin{lem}
    \label{lem:proj_grad}
    The optimizer $\rho^\star(x,g)$ of \eqref{eq:QP_reference_update} satisfies:
    \begin{align}
        \nabla_g V_g(g,r)^\nlT \rho^\star(x,g) \le - \|\rho^\star(x,g)\|^2.
        \label{eq:proj_grad_ineq}
    \end{align}
\end{lem}

\begin{proof}
    The QP solution $\rho^\star$ is the projection of $-\nabla_g V_g$ onto the convex set $\mathcal{H} = \{\rho \mid a_H^\nlT \rho \le b_H\}$.
    By the variational inequality of the projection, we have $(\rho^\star - (-\nabla_g V_g))^\nlT (\rho - \rho^\star) \ge 0$ for all $\rho \in \mathcal{H}$.
    Since $\rho = 0 \in \mathcal{H}$ (as shown in Proposition~\ref{prop:feasibility}), substituting $\rho=0$ yields
    $(\rho^\star + \nabla_g V_g)^\nlT (-\rho^\star) \ge 0$, which implies $- \|\rho^\star\|^2 - \nabla_g V_g^\nlT \rho^\star \ge 0$.
    Rearranging terms gives \eqref{eq:proj_grad_ineq}.
\end{proof}

\begin{lem}
    \label{lem:plant_bound}
    Let $V(x,g)$ be the reference-dependent Lyapunov function.
    Assume there exist constants $c_x, L_{xg} > 0$ such that:
    \begin{align}
        \frac{\partial V}{\partial x} f(x,\kappa(x,g)) & \le -c_x \|x-x_g\|^2, \label{eq:plant_bound_1} \\
        \left\| \frac{\partial V}{\partial g} \right\| & \le L_{xg} \|x-x_g\|. \label{eq:plant_bound_2}
    \end{align}
    Then, direct differentiation $\dot{V} = \frac{\partial V}{\partial x} f + \frac{\partial V}{\partial g}^\nlT \rho^\star$ and substitution of the bounds yield
    \begin{align}
        \dot{V}(x,g) \le -c_x \|x - x_g\|^2 + L_{xg} \|x-x_g\| \|\rho^\star\|
    \end{align}
    along the trajectories of the closed-loop system defined by \eqref{eq:augmented_system} and \eqref{eq:QP_solution}.
\end{lem}

\begin{thm}[Safety and convergence to stationary points]
    \label{thm:main_short_revised}
    Consider the closed-loop system \eqref{eq:augmented_system} with the update law \eqref{eq:QP_solution}. 
    Assume the potential $V_g(\cdot,r)$ and the reference-dependent Lyapunov function $V(x,g)$ are radially unbounded.
    Assume the constraint functions $h_i(x,g)$ are continuously differentiable and that the aggregated barrier $H$ satisfies $\nabla_g H(x,g)\neq0$ on $\partial\mathcal C$.
    Further assume $V$ satisfies the conditions of Lemma~\ref{lem:plant_bound}.
    If $(x(0),g(0))\in\mathcal C$, then $H(x(t),g(t))\ge 0$ for all $t\ge0$, and every trajectory converges to the largest invariant subset of
    \begin{align}
        \mathcal Z \coloneqq \{(x,g)\in\mathcal C \mid x=x_g,\ \rho^\star(x,g)=0\}.        
    \end{align}
\end{thm}
\begin{proof}
    By Proposition~\ref{prop:feasibility}, the QP constraint is always feasible on $\mathcal C$.
    Hence, the CBF condition \eqref{eq:H_CBF_condition} holds along the closed-loop trajectories,
    which implies forward invariance of $\mathcal C$.
    Therefore, $(x(t),g(t))\in\mathcal C$ for all $t\ge0$.

    Define the composite Lyapunov function
    \begin{align}
        W(x,g) \coloneqq V(x,g) + c V_g(g,r),        
    \end{align}
    with $c>0$ to be chosen.
    Using Lemmas~\ref{lem:proj_grad} and~\ref{lem:plant_bound}, we obtain
    \begin{align}
        \dot W
        \le
        -c_x \|x-x_g\|^2
        + L_{xg}\|x-x_g\|\|\rho^\star\|
        - c\|\rho^\star\|^2.        
    \end{align}
    Choosing $c>\frac{L_{xg}^2}{4c_x}$ renders the quadratic form in
    $(\|x-x_g\|,\|\rho^\star\|)$ negative definite.
    Hence, $\dot W \le 0$, with equality only if $x=x_g$ and $\rho^\star=0$.

    Therefore, $W(x(t),g(t))$ is nonincreasing and bounded below.
    Since both $V$ and $V_g$ are radially unbounded,
    boundedness of $W$ implies boundedness of $x(t)$ and $g(t)$.
    Thus, the trajectory remains in a compact forward-invariant set.

    Because $W$ is continuously differentiable and the trajectory remains in a compact set,
    LaSalle's invariance principle applies.
    Every trajectory converges to the largest invariant set contained in
    $\{\dot W=0\}$.
    From the characterization of $\dot W=0$,
    this set coincides with
    \begin{align}
        \mathcal Z
        =
        \{(x,g)\in\mathcal C \mid x=x_g,\ \rho^\star(x,g)=0\}.        
    \end{align}
    This concludes the proof.
\end{proof}

\begin{cor}
    Under the assumptions of Theorem~\ref{thm:main_short_revised}, assume in addition that:
    \begin{enumerate}
        \item Define $H_{\mathrm{ss}}(g) \coloneqq H(x_g, g)$ and the corresponding set $\mathcal V_H \coloneqq \{ g \mid H_{\mathrm{ss}}(g) \ge 0 \}$.
        Assume that $\mathcal V_H$ is convex and that $r \in \operatorname{int}\mathcal V_H$.
        \item The function $V_g(\cdot,r)$ is continuously differentiable and convex on $\mathcal V_H$, and admits $r$ as its unique minimizer over $\mathcal V_H$.
    \end{enumerate}
    Then every trajectory of the closed-loop system under \eqref{eq:QP_solution} satisfies $\lim_{t \to \infty} \|g(t)- r\|=0$ and $\lim_{t \to \infty} \|x(t)- x_r\| = 0$.
\end{cor}

\begin{proof}
    By Theorem~\ref{thm:main_short_revised}, every trajectory converges to the largest invariant subset of $\mathcal{Z}$. 
    Consider any point $(x_g,g)\in\mathcal Z$. Evaluating the aggregated barrier at steady state yields $H(x_g,g)=H_{\mathrm{ss}}(g)$. 
    Because $x_g$ is the equilibrium corresponding to the constant reference $g$, it minimizes the reference-dependent Lyapunov function; hence, $\nabla_x V(x_g,g) = 0$. 
    Recalling that each DSM satisfies $\nabla_x \Delta_i(x,g) = -\nabla_x V(x,g)$ and that $H$ is constructed by a smooth softmin of the DSMs and the steady-state terms, we obtain $\nabla_x H(x_g,g) = 0$. 
    Therefore, by the chain rule, the total derivative is
    \begin{align}
        \nabla_g H_{\mathrm{ss}}(g) &=
        \Big(\frac{\partial x_g}{\partial g}\Big)^\nlT \nabla_x H(x_g,g)
        + \frac{\partial H}{\partial g}(x_g,g) \\
        &=
        \frac{\partial H}{\partial g}(x_g,g),
    \end{align}
    meaning the partial derivative $\partial_g H(x_g,g)$ and the full gradient $\nabla_g H_{\mathrm{ss}}(g)$ coincide at steady state. 
    Moreover, since $f(x_g,\kappa(x_g,g))=0$, the CBF constraint at $(x_g,g)$ reduces to the affine inequality
    $-\nabla_g H_{\mathrm{ss}}(g)^\nlT \rho \le \alpha_H(H_{\mathrm{ss}}(g))$. 
    By Proposition~\ref{prop:feasibility}, this feasible set contains the origin whenever $H_{\mathrm{ss}}(g)\ge0$.

    The optimizer $\rho^\star$ is the Euclidean projection of $-\nabla_g V_g(g,r)$ onto this convex feasible set. 
    Thus $\rho^\star=0$ holds if and only if $-\nabla_g V_g(g,r)$ belongs to its normal cone at $\rho=0$. 
    We now show that this can occur only when $g=r$. 
    Suppose, for contradiction, that $g\neq r$ and $\rho^\star=0$.

    \textbf{Case 1:} $g\in\operatorname{int}\mathcal V_H$. 
    Then $H_{\mathrm{ss}}(g)>0$, so $\alpha_H(H_{\mathrm{ss}}(g))>0$ and the origin $\rho=0$ lies in the strict interior of the feasible set.
    Hence, the normal cone at $\rho=0$ reduces to $\{0\}$, which implies $\nabla_g V_g(g,r) = 0$. 
    Since $V_g(\cdot,r)$ is convex on $\mathcal V_H$ and admits $r$ as its unique minimizer, this contradicts $g\neq r$.

    \textbf{Case 2:} $g\in\partial\mathcal V_H$. 
    Then $H_{\mathrm{ss}}(g)=0$, and the feasible set reduces to the half-space 
    $\{ \rho \mid -\nabla_g H_{\mathrm{ss}}(g)^\nlT \rho \le 0 \}$,
    which coincides with the tangent cone of $\mathcal V_H$ at $g$. 
    The condition $\rho^\star=0$ implies $-\nabla_g V_g(g,r)\in N_{\mathcal V_H}(g)$, meaning
    $-\nabla_g V_g(g,r)^\nlT (y-g) \le 0$ for all $y\in\mathcal V_H$. 
    Taking $y=r$ yields $\nabla_g V_g(g,r)^\nlT (r-g) \ge 0$.
    However, by the convexity of $V_g(\cdot,r)$ and the unique minimum at $r$, we have
    $V_g(r,r) \ge V_g(g,r) + \nabla_g V_g(g,r)^\nlT (r-g)$. 
    Since $V_g(r,r) < V_g(g,r)$, this requires $\nabla_g V_g(g,r)^\nlT (r-g) < 0$, which contradicts the previous inequality.

    Having excluded both cases, we conclude the only point in $\mathcal Z$ satisfying $\rho^\star=0$ is $g=r$; hence, $\mathcal Z=\{(x_r,r)\}$. 
    Since $\mathcal Z=\{(x_r,r)\}$, Theorem~\ref{thm:main_short_revised} implies
    $\lim_{t\to\infty}\|g(t)-r\| = 0$ and $\lim_{t\to\infty}\|x(t)-x_r\| = 0$.
\end{proof}

\subsection{Discussion and Comparison}

\subsubsection{Comparison with CBF and ERG frameworks}
The proposed ERG-CBF approach combines reference governing with CBF-based safety enforcement.
Unlike standard CBF-QP filters that directly constrain the plant input and may suffer infeasibility under tight actuation limits or conflicting barriers \cite{amesControlBarrierFunctions2019,chen2021}, the proposed method encodes hard input bounds into DSMs and regulates the auxiliary reference. 
Since the trivial update $\rho=0$ remains admissible within the safe set (Proposition~\ref{prop:feasibility}), recursive feasibility of the reference QP is structurally guaranteed.
Compared to conventional ERG schemes relying on heuristic navigation fields and repulsion shaping in nonconvex environments \cite{8412335}, the update direction is obtained via a one-step closed-form projection.
This least-deviation correction adapts automatically to the local constraint geometry while preserving the Lyapunov-based transient guarantees of ERG methods.

\subsubsection{Parameter selection and implementation}
When the safety constraint is inactive, the proposed update reduces to the nominal gradient flow $\dot g = -\nabla_g V_g(g,r)$.
In contrast, conventional ERG schemes typically scale the navigation field by a gain and the safety margin, which may result in significantly larger reference velocities depending on the tuning.
As a consequence, in the proposed framework, the transient convergence rate is directly determined by the scaling (or gradient magnitude) of the potential function $V_g$.
Therefore, to achieve comparable convergence speed, $V_g$ should be chosen sufficiently steep.
The governor behavior depends on the potential $V_g$, the softmin parameters $\beta_H$ (used in $H$) and $\beta_\Delta$ (used in the DSM smoothing), and the CBF gain $\alpha_H$.
Larger $\beta_H$ and $\beta_\Delta$ values yield tighter approximations to the exact minimum but may worsen numerical conditioning; moderate values (tens to low hundreds) are typically sufficient in practice.

\section{SIMULATION RESULTS}
\label{sec:sim}

To demonstrate the effectiveness of the proposed CBF-based reference governor on nonlinear systems, we apply it to an $n$-degree-of-freedom ($n$-DOF) planar robotic manipulator.
Let $q \in \nlR^n$ denote the joint angles, $\tau \in \nlR^n$ the applied joint torques, and $X \coloneqq [q^\nlT~\dot{q}^\nlT]^\nlT \in \nlR^{2n}$ the state.
Assuming horizontal planar motion where the effects of gravity are neglected, the dynamics are given by $M(q)\ddot{q} + C(q,\dot{q})\dot{q} = \tau$, where $M(q)$ is the inertia matrix and $C(q,\dot{q})$ represents the Coriolis terms \cite{spong2020robot}.

A joint-space PD controller $\tau = -K_P (q - g) - K_D \dot{q}$ steers the system toward the auxiliary reference $g \in \nlR^n$.
The closed-loop system admits the reference-dependent Lyapunov function $V(X,g) \coloneqq \frac{1}{2}\dot{q}^\nlT M(q) \dot{q} + \frac{1}{2}(q - g)^\nlT K_P (q - g)$.
Exploiting the skew-symmetry of $\dot{M}(q) - 2C(q,\dot{q})$ (see Chapter 6 in \cite{spong2020robot}), its time derivative satisfies $\dot{V}(X,g) = -\dot{q}^\nlT K_D \dot{q} \le 0$ for all $X \in \nlR^{2n}$.

For whole-arm collision avoidance against a circular obstacle at $p_o \in \nlR^2$ with radius $R$, we discretize the $k$-th link into $N_k$ points $p_{k,j}(g) \coloneqq P_{k-1}(g) + \frac{j}{N_k} ( P_k(g) - P_{k-1}(g) )$, where $P_k(g)$ is the $k$-th joint position.
The Euclidean distance to the obstacle is $d_{k,j}(g) \coloneqq \|p_{k,j}(g) - p_o\|$.
Using a softmin approximation with the parameter $\beta > 0$, the smooth steady-state distance is formulated as $\tilde{d}_{\mathrm{arm}}(g, p_o) \coloneqq -\frac{1}{\beta} \log \big( \sum_{k=1}^n \sum_{j=1}^{N_k} \exp(-\beta d_{k,j}(g)) \big)$, ensuring steady-state admissibility via $h_{\mathrm{arm}}(g) \coloneqq \tilde{d}_{\mathrm{arm}}(g, p_o) - R \ge 0$.

Exploiting the Lipschitz continuity of forward kinematics map---that transforms joint angles into position and orientation of the last link of the robot---, the spatial error of any point on the manipulator is bounded by $L_{\mathrm{max}} \|q - g\|$, where the Lipschitz constant is bounded by $L_{\mathrm{max}} \coloneqq \sqrt{\sum_{m=1}^n (\sum_{i=m}^n l_i)^2}$.
Since the joint error satisfies $\|q - g\| \le \sqrt{2V(X,g)/\lambda_{\min}(K_P)}$, the worst-case spatial error is rigorously confined by the Lyapunov energy.
The closed-form DSM for transient safety is $\Delta_{\mathrm{arm}}(X,g) \coloneqq \Gamma_{\mathrm{arm}}^\ast(g) - V(X,g)$, with $\Gamma_{\mathrm{arm}}^\ast(g) \coloneqq \frac{\lambda_{\min}(K_P)}{2 L_{\mathrm{max}}^2} ( \max\{0, \tilde{d}_{\mathrm{arm}}(g, p_o) - R\} )^2$, which is aggregated into the softmin barrier function $H(X,g)$ as per \eqref{eq:H_softmin_alt}.

\begin{figure*}[t]
    \centering
    \begin{subfigure}[b]{0.32\linewidth}
        \centering
        \includegraphics[width=\linewidth]{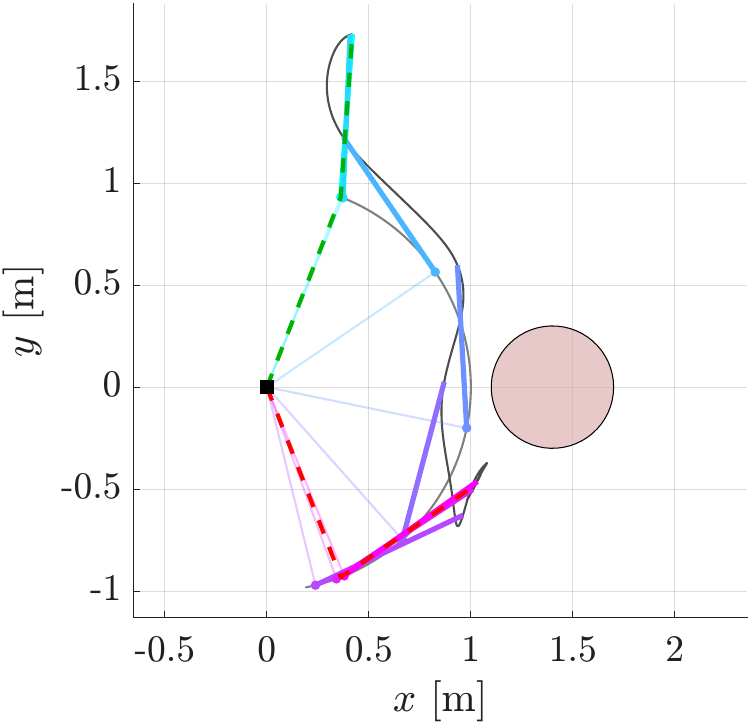}
        \caption{\footnotesize Workspace trajectory showing the manipulator safely avoiding the circular obstacle.}
        \label{fig:traj_workspace}
    \end{subfigure}
    \hfill
    \begin{subfigure}[b]{0.32\linewidth}
        \centering
        \includegraphics[width=\linewidth]{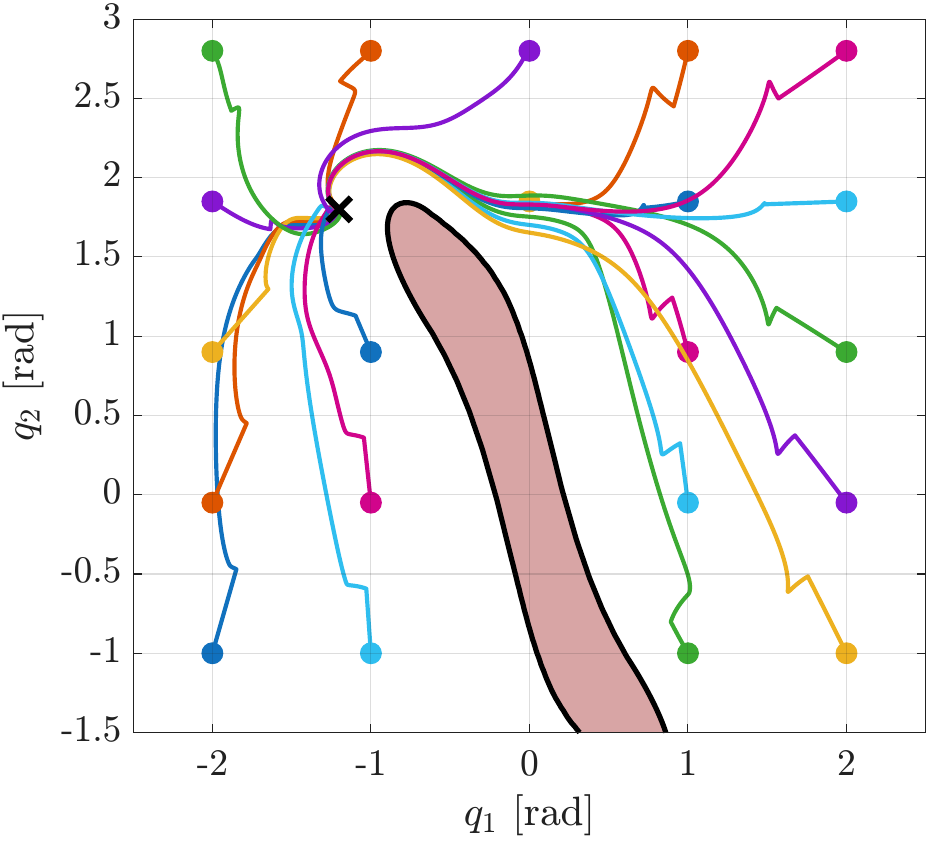}
        \caption{\footnotesize C-space trajectories from 20 initial configurations converging to the target.}
        \label{fig:traj_cspace}
    \end{subfigure}
    \hfill
    \begin{subfigure}[b]{0.32\linewidth}
        \centering
        \includegraphics[width=\linewidth]{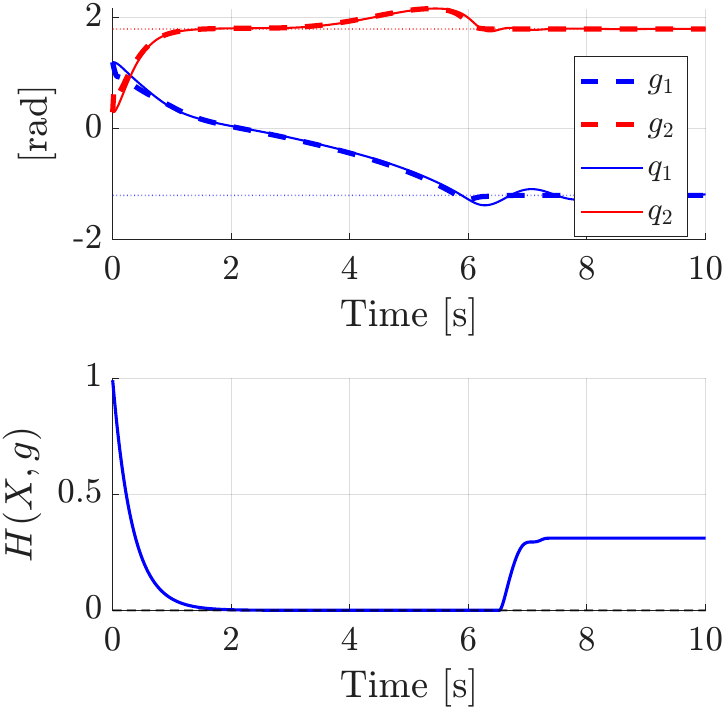}
        \caption{\footnotesize Time responses of the joint angles and the aggregated barrier function $H(X,g)$.}
        \label{fig:manipulator_time}
    \end{subfigure}
    \caption{\footnotesize Simulation results for the 2-DOF planar manipulator navigating around a static circular obstacle using the proposed ERG-CBF framework.}
    \label{fig:manipulator_results}
\end{figure*}

We simulate a $2$-DOF planar manipulator
($l_1 = \SI{1.0}{\meter}$, $l_2 = \SI{0.8}{\meter}$, $m_1 = \SI{2.0}{\kilogram}$, $m_2 = \SI{1.0}{\kilogram}$)
with a circular obstacle at
$p_o = [1.4~0]^\nlT~\si{\meter}$ and $R = \SI{0.30}{\meter}$.
The controller gains are $K_P = 50 I_2$, $K_D = 3 I_2$, and the ERG-CBF parameters are $P = 15 I_2$, $\alpha_H(s) = 3s$, $\beta = 100$, and $N_k =  5$.
In the numerical implementation, standard log-sum-exp scaling was applied to ensure numerical stability.

The results are shown in Fig.~\ref{fig:manipulator_results}.
Fig.~\ref{fig:traj_workspace} shows the workspace motion from the initial configuration $q(0) = [1.2~0.3]^\nlT$\,rad (green dashed) to the target (red dashed), where the gray curves trace the elbow and tip trajectories, and the colored lines depict the second-link snapshots, progressing from cyan (initial) to magenta (final) over time.
The proposed method drives the manipulator to the target without any part of the arm contacting the obstacle (shaded circle).
In Fig.~\ref{fig:traj_cspace}, we report the results of 20 simulations obtained varying the initial configuration of the robot. Each curve represents the governor trajectory $g(t)$ starting from a distinct initial configuration ($\bullet$), all converging to the target $r$ ($\times$); the shaded region denotes the obstacle in the configuration space (typically referred to as C-space \cite{spong2020robot}).
Using the closed-form projection along $\nabla_g H$, the governor automatically adapts the update direction to the local constraint geometry, eliminating the need for manual navigation-field design.
Fig.~\ref{fig:manipulator_time} shows the time responses of the joint angles and the aggregated barrier $H(X,g)$, confirming constraint satisfaction throughout the motion.

\section{CONCLUSION}
\label{sec:conclusion}
In this letter, we proposed a novel Explicit Reference Governor framework that integrates Control Barrier Functions with Dynamic Safety Margins to handle multiple state and input constraints.
By employing the softmin function, we aggregated multiple constraints into a single smooth barrier function, enabling the derivation of a closed-form, optimization-based reference update law.
A key feature of the proposed method is the structural guarantee of recursive feasibility, which addresses a common limitation of standard CBF-QP approaches.
Numerical simulations on an $n$-DOF planar robotic manipulator demonstrated that the proposed framework effectively ensures safety and convergence in the presence of obstacles and input saturation, automatically adapting the reference trajectory without the need for manual and heuristic design of navigation fields.

\end{document}